\documentclass[letterpaper,11pt]{article}
\usepackage[margin=1in]{geometry}
\usepackage{amsmath, amsfonts, amssymb, amsthm}
\usepackage{thm-restate}
\usepackage{bbm,url}
\usepackage[colorlinks=true,urlcolor=black,linkcolor=black,citecolor=black]{hyperref}
\usepackage{graphicx}
\usepackage[ruled,linesnumbered]{algorithm2e}
\usepackage{color}
\usepackage{subcaption}
\usepackage{hhline}
\usepackage{pifont}

\newtheorem{theorem}{Theorem}
\newtheorem*{theorem*}{Theorem}
\newtheorem{corollary}[theorem]{Corollary}
\newtheorem{lemma}[theorem]{Lemma}
\newtheorem{claim}{Claim}
\newtheorem{remark}[theorem]{Remark}
\newtheorem{observation}[theorem]{Observation}
\newtheorem{example}[theorem]{Example}

\theoremstyle{definition}

\newcommand{\hide}[1]{}

\newtheorem{question}{Question}

\newcommand{\s}[1]{\mathsf{#1}}
\newcommand{\NW}{\s{NW}}

\newcommand{\prop}{\s{Prop}}
\newcommand{\rrs}{\s{RRS}}

\newcommand{\NP}{\s{NP}}
\newcommand{\A}{\mathcal{A}}
\newcommand{\G}{\mathcal{G}}
\newcommand{\I}{\mathcal{I}}
\newcommand{\x}{\mathbf{x}}
\newcommand{\y}{\mathbf{y}}

\newcommand{\PGNSW}{\s{PGNW}}
\newcommand{\ERSP}{\s{ERSP}}
\newcommand{\N}{\mathbb N}
\newcommand{\Z}{\mathbb Z}

\newcommand{\xmark}{\ding{55}}

\newcommand{\private}{\s{PrivateGoods}}
\newcommand{\privatemnw}{\s{PrivateMNW}}
\newcommand{\privatelex}{\s{PrivateLex}}
\newcommand{\public}{\s{PublicGoods}}
\newcommand{\publicmnw}{\s{PublicMNW}}
\newcommand{\publiclex}{\s{PublicLex}}
\newcommand{\decision}{\s{PublicDecisions}}
\newcommand{\decisionmnw}{\s{DecisionMNW}}
\newcommand{\decisionlex}{\s{DecisionLex}}

\title{On Fair and Efficient Allocations of Indivisible Public Goods\thanks{Supported by NSF Grant CCF-1942321 (CAREER)}}

\author{Jugal Garg \footnote{University of Illinois at Urbana-Champaign, USA} \\
\texttt{\small jugal@illinois.edu} 
\and
Pooja Kulkarni\footnote{University of Illinois at Urbana-Champaign, USA}\\
\texttt{\small poojark2@illinois.edu}
\and
Aniket Murhekar\footnote{University of Illinois at Urbana-Champaign, USA}\\
\texttt{\small aniket2@illinois.edu}
}

\date{}
\begin{document}

\maketitle

\begin{abstract}
We study fair allocation of indivisible public goods subject to cardinality (budget) constraints. In this model, we have $n$ agents and $m$ available public goods, and we want to select $k \leq m$ goods in a fair and efficient manner. We first establish fundamental connections between the models of private goods, public goods, and public decision making by presenting polynomial-time reductions for the popular solution concepts of maximum Nash welfare (MNW) and leximin. These mechanisms are known to provide remarkable fairness and efficiency guarantees in private goods and public decision making settings. We show that they retain these desirable properties even in the public goods case. We prove that MNW allocations provide fairness guarantees of Proportionality up to one good (Prop1), $1/n$ approximation to Round Robin Share (RRS), and the efficiency guarantee of Pareto Optimality (PO). 
Further, we show that the problems of finding MNW or leximin-optimal allocations are $\NP$-hard, even in the case of constantly many agents, or binary valuations. This is in sharp contrast to the private goods setting that admits polynomial-time algorithms under binary valuations. We also design pseudo-polynomial time algorithms for computing an exact MNW or leximin-optimal allocation for the cases of (i) constantly many agents, and (ii) constantly many goods with additive valuations. We also present an $O(n)$-factor approximation algorithm for MNW which also satisfies RRS, Prop1, and $1/2$-Prop.
\end{abstract}

\section{Introduction}
The problem of fair division was formally introduced by Steinhaus~\cite{steinhaus1948problem}, and has since been extensively studied in economics and computer science \cite{Brams1996FairD,Moulin2003}. Recent work has focused on the problem of fair and efficient allocation of indivisible \textit{private} goods. We label this setting as the $\private$ model. Here, goods have to be partitioned among agents, and a good provides utility only to the agent who owns it. However, goods are not always private, and may provide utility to multiple agents simultaneously, e.g., books in a public library. The fair and efficient allocation of such \emph{indivisible public goods} is an important problem.

In this paper we study the setting of $\public$, where a set of $n$ agents have to select a set of at most $k$ goods from a set of $m$ given goods. This simple cardinality constraint models several real world scenarios. While previous work has largely focused on the $k < n$ case, e.g., for voting and committee selection \cite{AzizVoting,committeemunagala}, there is much less work available for the case of $k \geq n$. This setting is important in its own right. We present a few compelling examples.

\begin{example}\label{ex:library}
A \textit{public library} wants to buy $k$ books that adhere to preferences of $n$ people who might use the library. Clearly, the number of books has to be much greater than the number of people using the library, hence $k \gg n$.
\end{example}

\begin{example}\label{ex:movies}
A family (or a group of $n$ friends) wants to decide on a list of $k$ movies to watch together for a few months. Here too, $k > n$. Another example of the same flavor is a committee tasked with inviting speakers at a year-long weekly seminar.
\end{example}

\begin{example}\label{ex:search}
Another important example is that of \textit{diverse search results} for a query. Given a query (say of ``computer scientist images'') on a database, we would like to output $k$ search results which reflect diversity in terms of $n$ specified features (like ``gender, race and nationality''). Once again, $k \ge n$.
\end{example} 

A related setting $\decision$ of public decision making~\cite{PublicDecision} models the scenario in which $n$ agents are faced with $m$ issues with multiple alternatives per issue, and they must arrive at a decision on each issue. Conitzer \textit{et al.}~\cite{PublicDecision} showed that this model subsumes the $\private$ setting. 

\paragraph{Connections between the models.} A central question motivating this work is:

\begin{question}\label{qn:connections}
Can we establish \textit{fundamental connections} between the three models $\private$, $\public$, and $\decision$?
\end{question}

To answer this question, we first describe two well-studied solution concepts for allocating goods in the $\private$ and $\decision$ models, namely the \textit{maximum Nash welfare} (MNW) and \textit{leximin} mechanisms. These mechanisms have been shown to produce allocations that are fair and efficient in the models of $\private$ and $\decision$. The MNW mechanism returns an allocation that maximizes the geometric mean of agents' utilities, and the leximin mechanism returns an allocation that maximizes the minimum utility, and subject to this, maximizes the second minimum utility, and so on. We label the problems of computing the Nash welfare maximizing (resp. leximin optimal) allocation in the three models as $\privatemnw, \publicmnw, \decisionmnw$ (resp. $\privatelex, \publiclex, \decisionlex$).

We answer Question~\ref{qn:connections} positively by presenting novel polynomial-time reductions from the model of $\private$ to $\public$, and from $\public$ to $\decision$ for the problem of computing a Nash welfare maximizing allocation. 
\begin{equation}\label{eq:reduction1}
\boxed{ \privatemnw \le \publicmnw \le \decisionmnw}
\end{equation}
More notably, these reductions also work for the MNW problem when restricted to binary valuations. Apart from establishing fundamental connections between these models, our reductions also determine the complexity of the MNW problem, as we detail below. We also develop similar reductions between the models for the leximin mechanism, showing:
\begin{equation}\label{eq:reduction2}
\boxed{ \privatelex \le \publiclex \le \decisionlex }
\end{equation}

\paragraph{Fairness and efficiency considerations.} We next describe the fairness and efficiency properties that the MNW and leximin mechanisms have been shown to satisfy in the $\private$ and $\decision$ models. 

The standard notion of economic efficiency is Pareto-optimality (PO). An allocation is said to be PO if no other allocation makes an agent better off without making anyone worse off. The classical fairness notion of \textit{proportionality} requires that every agent gets her \textit{proportional value}, i.e., $1/n$-fraction of the maximum value she can obtain in any allocation. However, proportional allocations are not guaranteed to exist.\footnote{Consider for example, two agents $A$ and $B$ and six public goods $\{g_1, g_2, g_3, g_4, g_5, g_6\}$. Agent $A$ has value $1$ for $g_1, g_2, g_3$ and $B$ has value $1$ for $g_4, g_5, g_6$. All other valuations are $0$. Suppose we want to select three of these goods. The proportional share of both agents is $1.5$. However, in any allocation, the value of at least one agent is at most $1$, implying that proportional allocations need not exist.} Hence, we study the notion of Proportionality up to one good (Prop1) for $\public$. We say an allocation is Prop1 if for every agent $i$ who does not get her proportional value, $i$ gets her proportional value after swapping some unselected good with a selected one. For $\private$ and $\decision$, Prop1 is defined similarly -- in the former, an agent is given an additional good \cite{barman2019prop1,McGlaughlinG20}; and in the latter, an agent is allowed to change the decision on a single issue \cite{PublicDecision}. While Prop1 is an individual fairness notion, it is still important for allocating public goods. For instance, in Example~\ref{ex:library}, we want allocations in which every agent has some books that cater to her taste, even if her taste differs from the rest of the agents. Likewise, in Example~\ref{ex:movies}, a fair selection of movies must ensure that there are some movies every member can enjoy. We also consider the fairness notion of Round-Robin Share (RRS)~\cite{PublicDecision}, which demands that each agent $i$ receives at least the utility which she would get if agents were allowed to pick goods in a round-robin fashion, with $i$ picking last.

In the $\private$ and $\decision$ models, an MNW allocation satisfies Prop1 in conjunction with PO ~\cite{caragiannis2019unreasonable,PublicDecision}. Similarly in both these models, the leximin-optimal allocation satisfies RRS and PO~\cite{PublicDecision}. It is therefore natural to ask:

\begin{question}
What guarantee of fairness and efficiency do the MNW and leximin mechanisms provide in the $\public$ model?
\end{question}

Answering this question, we show that an MNW allocation satisfies Prop1, $1/n$-approximation to RRS, and is PO. Further, a leximin-optimal allocation satisfies RRS, Prop1 and PO.

\paragraph{Complexity of computing MNW and leximin-optimal allocations.}
Given the desirable fairness and efficiency properties of these mechanisms, we investigate the complexity of computing MNW and leximin-optimal allocations in the $\public$ model. It is known that $\privatemnw$ is $\s{APX}$-hard \cite{lee2017apx} (hard to approximate) and $\decisionmnw$ \cite{PublicDecision} is $\s{NP}$-hard.  Likewise, $\privatelex$ too is $\s{NP}$-hard~\cite{bezakova2005leximin}. Therefore, we ask:

\begin{question}
What is the complexity of $\publicmnw$ and $\publiclex$?
\end{question}

Since $\privatemnw$ and $\privatelex$ are known to be $\s{NP}$-hard, our reductions \eqref{eq:reduction1} and \eqref{eq:reduction2} immediately show that $\publicmnw$ and $\publiclex$ are $\s{NP}$-hard. However, we show stronger results that $\publicmnw$ and $\publiclex$ remain $\s{NP}$-hard even when the valuations are binary. These results are in stark contrast to the $\private$ case, which admits polynomial-time algorithms for binary valuations \cite{NSWbin2,freeman2019eqxpo}. Further, our reductions between $\public$ and $\decision$ also directly enable us to show NP-hardness of $\decisionmnw$ and $\decisionlex$. Moreover, a feature of our reductions (Observation~\ref{rem:reductionvalues}) enables us to shows that $\decisionmnw$ is $\s{NP}$-hard even for binary valuations, highlighting the utility of our reductions. We also show that $\publicmnw$ and $\publiclex$ remain $\s{NP}$-hard even when there are only two agents. We note that for the case of two agents, the $\s{NP}$-hardness of $\privatemnw$ and $\privatelex$ does not imply $\s{NP}$-hardness of $\publicmnw$ and $\publiclex$ because our reductions between the models do not preserve the number of agents. We summarize our results in Table~\ref{tab:complexity}.

\begin{table}[t]
\centering
\begin{tabular}{|c|c|c|c|}\hline
Problem & $\private$ & $\public$ & $\decision$ \\\hhline{|=|=|=|=|}
MNW $\{0,1\}$ valuations & $\s{P}$~\cite{NSWbin1,NSWbin2} & $\s{NP}$-hard (Theorem~\ref{thm:mnw-main})& $\s{NP}$-hard (Corollary~\ref{cor:pdmhardness}) \\\hline
Leximin $\{0,1\}$ valuations & $\s{P}$~\cite{NSWbin1,NSWbin2} & $\s{NP}$-hard (Theorem~\ref{thm:np-hard-lm-bin})& ? \\\hline
MNW two agents & $\s{NP}$-hard & $\s{NP}$-hard (Theorem~\ref{thm:np-const})& ? \\\hline
Leximin two agents & $\s{NP}$-hard & $\s{NP}$-hard (Theorem~\ref{thm:hard-lm-const})& ? \\\hline
\end{tabular}
\caption{Complexity of computing MNW and leximin-optimal allocations}
\label{tab:complexity}
\end{table}

In light of the above computational hardness, we turn to approximation algorithms and exact algorithms for special cases. We design a polynomial-time algorithm that returns an allocation which approximates the MNW to a $O(n)$-factor when $k\ge n$, and is also Prop1 and satisfies RRS. Finally, we obtain pseudo-polynomial time algorithms for computing MNW and leximin-optimal allocations for constant $n$. These are essentially tight in light of the $\s{NP}$-hardness for constant $n$.

\subsection{Other related work}
\smallskip\noindent\textbf{Maximum Nash welfare.} The problem of approximating maximum Nash welfare for private goods is well-studied, see e.g.,~\cite{cole2015approximating,barman2018finding,chaudhary2018approxnswindiv,GargHV21}. \cite{ShahPublic} showed that the MNW problem is NP-hard for allocating public goods subject to matroid or packing constraints.
It has also been studied in the context of voting, or multi-winner elections~\cite{PositionalSD}. Fluschnik \textit{et al.}~\cite{fluschnik2019fair} studied the fair multi-agent knapsack problem, wherein each good has an associated budget, and a set of goods is to be selected subject to a budget constraint. In this context, they studied the objective of maximizing the geometric mean of $(1+u_i)$ where $u_i$ is the utility of the $i^{th}$ agent. They showed that maximizing this objective is $\s{NP}$-hard, even for binary valuations or constantly many agents with equal budgets and presented a pseudo-polynomial time algorithm for constant $n$.

\smallskip\noindent\textbf{Leximin.} Leximin was developed as a fairness notion in itself \cite{rawls2009theory}. Plaut and Roughgarden \cite{plaut2020almost} showed that for private goods, leximin can be used to construct allocations that are envy-free up to any good. Freeman \textit{et al.}~\cite{freeman2019eqxpo} showed that in the $\private$ model the MNW and leximin-optimal allocations coincide when valuations are binary.

\smallskip\noindent\textbf{Core.}
Core is a strong property that enforces both PO and proportionality-like fairness guarantees for all subsets of agents. It is well-studied in many settings, including game theory and computer science \cite{scarf1967core,kunjir2017robus}. The core of indivisible public goods might be empty. Fain \textit{et al.}~\cite{ShahPublic} proved that under matroid constraints, a $2$-additive approximation to core exists. On an individual fairness level, 1-additive core is weaker than Prop1~\cite{ShahPublic}.

\smallskip\noindent\textbf{Participatory Budgeting.} The participatory budgeting problem \cite{aziz2018pb,aziz2020participatory} consists of a set of $n$ agents (or voters), and a set of $k$ projects that require funds, a total available budget, and the preferences of the voters over the projects. The problem is to allocate the budget in a fair and efficient manner. Here typically $k\ll n$. Fain \textit{et al.}~\cite{partbudget} showed that the fractional core outcome is polynomial-time computable. This could be modeled as a public goods problem with goods as the projects.

\smallskip\noindent\textbf{Voting and Committee Selection.} These settings involve selecting a set of $k$ members from a set of $m$ candidates based on the preferences of $n$ agents. Usually, here $k \ll n$ and the fairness notions studied are group fairness like Justified Representation \cite{AzizVoting}, and a core-like notion called \textit{stability} \cite{committeemunagala}.

\section{Notation and Preliminaries}
\paragraph{Problem setting.} For $t\in \mathbb{N}$, let $[t]$ denote $\{1,\dots,t\}$. An instance of the $\public$ allocation problem is given by a tuple $\I = (\A,\G,k,\{v_i\}_{i\in\A})$ of a set $\A = [n]$ of $n\in\N$ agents, a set $\G = [m]$ of $m\in\N$ public goods, an integer $0 \le k \le m$, and a set of valuation functions $\{v_i\}_{i\in \A}$, one per agent, where each $v_i : 2^\G \rightarrow \Z_{\ge 0}$. Unless specified, we assume that $k\ge n$. For a subset of goods $S\subseteq\G$, $v_i(S)$ denotes the utility agent $i$ derives from the goods in $S$. Unless specified, we assume the valuations are \textit{additive}. In this case, each $v_i$ is specified by $m$ non-negative integers $\{v_{ij}\}_{j\in\G}$, where $v_{ij}$ denotes the value of agent $i$ for good $j$. Then for $S \subseteq \G$, $v_i(S) = \sum_{j\in S} v_{ij}$. We assume without loss of generality that for every agent $i$, there is at least one good $j$ with $v_{ij} > 0$. For brevity, we write $v_i(g_1,\ldots,g_r)$ in place of $v_i(\{g_1,\ldots,g_r\})$ for a set $\{g_1,\ldots,g_r\}\subseteq\G$. An \textit{allocation} is a subset $\x\subseteq\G$ of goods which satisfies the cardinality constraint $|\x|\le k$.

\paragraph{Nash welfare.} The Nash welfare (NW) of an allocation $\x$ is given by $\NW(\x) = ( \prod_{i\in \A} v_i(\x))^{1/n}.$ An allocation with the maximum NW is called an MNW allocation or a Nash optimal allocation.\footnote{If the NW is 0 for all allocations, MNW allocations are defined as those which give non-zero utility to maximum number of agents, and then maximize the product of utilities for those agents. Note if $k \geq n$, every agent positively values at least one good and thus MNW $> 0$.} We also refer to the product of the agents' utilities as the \emph{Nash product}. An allocation $\x$ \textit{approximates} MNW to a factor of $\alpha$ if $\NW(\x) \ge \alpha \cdot \NW(\x^*)$, where $\x^*$ is an MNW allocation.

\paragraph{Leximin.} 
Given an allocation $\x$, let $\hat{\x}$ denote the vector of agent's utilities under $\x$, sorted in non-decreasing order. For two allocations $\x,\y$, we say $\x$ \emph{leximin-dominates} $\y$ if there exists $i \in [n]$ such that $\hat{\x}_i > \hat{\y}_i$ and $\forall j < i, \hat{\x}_j = \hat{\y}_j$. An allocation is leximin-optimal if no other allocation leximin-dominates it.

\paragraph{Fairness notions.}
We now discuss fairness notions for the $\public$ setting. The \textit{proportional share} of an agent $i$, denoted by $\prop_i$ is a $1/n$-share of the maximum value she can obtain from any allocation. Formally: 
\[\prop_i = \frac{1}{n} \cdot \max_{\x\subseteq\G, |\x|\le k}  v_i(\x).\]
The round-robin share of agent $i$, denoted by $\rrs_i$, is the minimum value an agent can be guaranteed if the agents pick $k$ goods in a round-robin fashion, with $i$ picking last. Therefore, this value equals the maximum value of any $\lfloor k/n \rfloor$ sized subset. Formally: 
\[\rrs_i = \max_{\x\subseteq\G, |\x|\le \lfloor k/n \rfloor}  v_i(\x).\]
For $\alpha\in (0,1]$, an allocation $\x$ is said to satisfy:
\begin{enumerate}
\item $\alpha$-Proportionality ($\alpha$-Prop) if $\forall i\in\A$, $v_i(\x) \ge \alpha\prop_i$;
\item $\alpha$-Proportionality up to one good ($\alpha$-Prop1) if $\forall i\in\A$, $\exists g\in\x, g'\in \G$, such that $v_i((\x\setminus g)\cup g') \ge \alpha\prop_i,$
\item $\alpha$-$\s{RRS}$ if for all agents $i\in\A$, $v_i(\x) \ge \alpha\rrs_i$.
\end{enumerate}
Due to the cardinality constraints in the $\public$ model, the notion of Prop1 requires that for every agent, there is a way to swap one preferred unpicked good with one picked good, after which the agent gets her proportional share. Since Prop1 in $\private$ requires only giving an extra good, this makes the definition of Prop1 in $\public$ slightly more demanding than that in $\private$.

\paragraph{Pareto-optimality.} An allocation $\y$ is said to Pareto-dominate an allocation $\x$ if for all agents $i\in\A$, $v_i(\y) \geq v_i(\x)$, with at least one of the inequalities being strict. We say $\x$ is Pareto-optimal (PO) if there is no allocation that Pareto-dominates $\x$.

\paragraph{Related models.}
\begin{enumerate}
\item $\private$. The classic problem of \textit{private goods allocation} concerns partitioning a set of goods $\G$ among the set $\A$ of agents. Thus, a feasible allocation $\x$ is an $n$-partition $(\x_1, \dots, \x_n)$ of $\G$, where agent $i$ is allotted $\x_i\subseteq \G$, and derives utility $v_i(\x_i)$ only from $\x_i$.
\item $\decision$. In this model, a set $\A$ of agents are required to make \textit{decisions} on a set $\G$ of issues. Each issue $j \in \G$ has a set $\G_j$ of $k_j$ alternatives, given by $\G_j := \{(j,1), (j,2),\dots,(j,k_j)\}$. A feasible allocation or outcome $\x = (x_1,\dots,x_m)$ comprises of $m$ decisions, where $x_j\in [k_j]$ is the decision made on issue $j$. Assuming the valuations are additive, each agent has a value $v_i(j, \ell)$ for the $\ell^{th}$ alternative of issue $j$. The valuation of the agent for the outcome $\x$ is then $v_i(\x) = \sum_{j \in \G} v_{i}(j,x_j)$. 
\end{enumerate}

\section{Relating the models}\label{sec:relations}
We first show rigorous mathematical connections between the $\private, \public$ and $\decision$ models w.r.t. computing optimal MNW and leximin allocations. 
\begin{theorem} \label{thm:mnwpgtopd}
$\publicmnw$ polynomial-time reduces to $\decisionmnw$.
\end{theorem}
\begin{proof}
Let $\I = (\A,\G,k,\{v_i\}_{i\in \A})$ be an instance of the $\public$ model. 
For $k=m$, the MNW problem is trivial, since we can select all the $m$ goods. 
For $n \leq k < m$, we can construct an instance $\I' = (\A',\G',\{\G_j\}_{j\in \G'}\{v'_i\}_{i\in \A'})$ of $\decision$ from $\I$ in polynomial time, such that given an MNW allocation of $\I'$, we can compute an MNW allocation of $\I$ in polynomial time.
Let $V = \max_{i,j} v_{ij}$. We create $m$ public issues: corresponding to each good $j\in \G$, we create an issue $j$ with two alternatives $(j,1)$ and $(j,2)$. That is, $\G' = [m]$, and $\G_j = \{(j,1),(j,2)\}$ for $j\in \G'$. We create $\A'=[n+mT]$, where $ T = \lceil 2mn\log mV \rceil$. The first $n$ agents here correspond to the $n$ agents in $\I$. The last $mT$ agents are of two types: $kT$ agents $\{n+1,\dots, n+kT\}$ of type $A$, and $(m-k)T$ agents $\{n+kT+1,\dots,n+mT\}$ of type $B$. The valuations are as follows: each agent $i\in[n]$ values alternative `$1$' of the issue $j\in \G'$ at $v_{ij}$, the agents of type $A$ value only alternative `$1$', agents of type $B$ value only alternative `$2$'. Formally, for $i\in \A'$, and an alternative $(j,c)$ of the issue $j \in \G'$, where $c\in \{1,2\}$:
\[
v'_i(j,c) = \begin{cases}
v_{ij}, \text{ if } c=1 \text{ and } i\in [n]; \\
1, \text{ if } n < i \le n+kT \text{ and } c = 1;\\
1, \text{ if } n+kT < i \le n+mT \text{ and } c = 2;\\
0, \text{ otherwise. }
\end{cases}
\]

Let $\x'$ be an allocation for the instance $\I'$. For $c\in\{1,2\}$, let $S_c$ be the set of issues $j$ with decision $c$ in $\x'$. That is, $S_c = \{j\in [m] : \x'_j = c\}$. Let $k' = |S_1|$. Then we have:

\[\NW(\x') = \bigg(\prod_{i\in[n]} v'_i(\x') \cdot (k')^{kT} \cdot (m-k')^{(m-k)T} \bigg)^{\frac{1}{n+mT}}.\]

We now relate $\x'$ to the $\public$ instance $\I$. The decision $(j,1)$ corresponds to selecting the public good $j$. Let $\x = S_1 \subseteq \G$ be the corresponding set of public goods. Then for any $i\in [n]$ we have that $v_i(\x) = v'_i(\x')$, since $v'_i(j,2) = 0$ for every $j\in [m]$. Thus:
\begin{equation}\label{eqn:pg-pdm-nsw}
\NW(\x') = \big(\NW(\x)^n \cdot (k')^{kT} \cdot (m-k')^{(m-k)T} \big)^{\frac{1}{n+mT}}.
\end{equation}

We now have to prove that $\x$ satisfies $\vert \x \vert \le k$. Let $W_\ell$ be the Nash product of any MNW allocation for the $\public$ instance $\I_\ell = (\A,\G,\ell,\{v_i\}_{i\in \A})$, $0\le \ell \le m$. Clearly, $0 = W_0 \le W_1 \le \dots W_m \le (mV)^n$. As $k \geq n$, $W_k \ge 1$, since we assume every agent has at least one good that she values positively. Define $g : [m] \rightarrow \mathbb{Z}$, as $g(a) = a^k (m-a)^{m-k}$. Then if $\x'$ is an MNW allocation for $\I'$,  \eqref{eqn:pg-pdm-nsw} becomes:
\begin{equation}\label{eqn:pg-pdm-nsw-2}
\NW(\x') = (W_{k'}\cdot g(k')^T)^{1/(n+mT)}.
\end{equation}

Let $G_1$ and $G_2$ denote the largest and second-largest values that $g$ attains over its domain. We observe that $g$ increases in $[0,k]$, and decreases in $[k,m]$. Hence, $G_1 = g(k)$ implying:
\begin{equation*}
\begin{aligned}
G_1 &= k^k(m-k)^{m-k};
G_2 = \max(g(k-1),g(k+1)).
\end{aligned}
\end{equation*}
\noindent We now show that that:
\begin{restatable}{claim}{clmpgtopgd}\label{clm:pubgoodstopubdec}
$G_1^T > W_m \cdot G_2^T$.
\end{restatable}
\begin{proof}
Recall that $W_\ell$ denotes the Nash product of any MNW allocation for the $\public$ instance $\I_\ell = (\A,\G,\ell,\{v_i\}_{i\in \A})$, for $0\le \ell \le m$. We have $0 = W_0 \le W_1 \le \dots W_m \le (mV)^n$, and we assume $W_k \ge 1$. Recall that function $g : [m] 
\rightarrow \mathbb{Z}$, was defined as $g(a) = a^k (m-a)^{m-k}$. 

Let $G_1$ and $G_2$ denote the largest and second-largest values that $g$ attains over its domain. We observe that $g$ increases in $[0,k]$, and decreases in $[k,m]$. Hence:
\begin{equation*}
\begin{aligned}
G_1 &= g(k) = k^k(m-k)^{m-k}. \\
G_2 &= \max(g(k-1),g(k+1)).
\end{aligned}
\end{equation*}
Now observe that for $k\in [m]\setminus\{0,1,m\}$:
\begin{equation*}
\begin{aligned}
\log g(k) - \log g (k-1) &= k(\log k - \log(k-1)) + (m-k)(\log(m-k) \\
&\,\,- \log(m-k+1)), \\
&> k\cdot\frac{1}{k-\frac{1}{2}}+ (m-k)\cdot\frac{-1}{m-k} \ge \frac{1}{2k-1} \ge \frac{1}{2m},
\end{aligned}
\end{equation*}
and for $k\in[m]\setminus\{0,m-1,m\}$:
\begin{equation*}
\begin{aligned}
\log g(k) - \log g(k+1) &= k(\log k - \log(k+1)) + (m-k)(\log(m-k)\\
&\,\,- \log(m-k-1)), \\
&> k\cdot\frac{-1}{k}+ (m-k)\cdot\frac{1}{m-k-\frac{1}{2}}, \\
&\ge \frac{1}{2(m-k)-1} \ge \frac{1}{2m},
\end{aligned}
\end{equation*}
using standard properties of logarithms. Thus:
\[\log G_1 - \log G_2 > \frac{1}{2m}.\]
Then we have by recalling that $T = 2mn\log mV$,
\[T (\log G_1 - \log G_2) > 2mn\log mV\cdot\frac{1}{2m} \ge \log W_m,
\]
which gives:
\[G_1^T > W_m \cdot G_2^T, \]
as required.
Lastly, we consider the cases of $k=1$ and $k=m-1$. In both cases, $T(\log G_1 -\log G_2) = T[(m-1)\log(m-1) - \log 2-(m-1)\log(m-2)] > 2mn\log mV\frac{1}{2m} \ge \log W_m$, which gives $G_1^T > W_m G_2^T$, as claimed. 
\end{proof}

\noindent Using Claim~\ref{clm:pubgoodstopubdec} we have for all $k'\in[m]\setminus\{k\}$:
\[W_k \cdot g(k)^T \ge G_1^T > W_m \cdot G_2^T \ge W_{k'} \cdot g(k')^T, \]

Hence, the quantity $W_{k'}\cdot g(k')^T$ is maximized when $k' = k$. Recalling \eqref{eqn:pg-pdm-nsw-2}, we conclude that for the MNW allocation $\x'$ of $\I'$, the corresponding set $\x$ has cardinality exactly $k$. Further $\x$ also maximizes the NW among all allocations of the instance $\I$ satisfying this cardinality constraint. Thus, $\x$ in fact is an MNW allocation for $\I$. Finally, it is clear that this is a polynomial time reduction.
\end{proof}

We next relate the MNW problem in the $\private$ model with the $\public$ model. 

\begin{restatable}{theorem}{thmmnwprivtopublic}\label{thm:mnw-priv-to-public}
$\privatemnw$ polynomial-time reduces to $\publicmnw$.
\end{restatable}
\begin{proof}
Let $\I=(\A=[n],\G=[m],V)$ be a $\private$ instance, using which we create a $\public$ instance $I'$ as follows. We create $n+2m$ agents, i.e. $\A' = [n+2m]$. The first $n$ agents correspond to the $n$ agents in $\I$. The last $2m$ are dummy agents. We create $n \cdot m$ public goods: for each good $j\in [m]$, we create a set of $n$ copies $S_j = \{j_1,j_2,\dots,j_n\}$, $\G' = \bigcup_{j\in \G} S_j$. We set $k = m$. The valuations for $i\in \A'$, $j_\ell \in \G'$ are:
\[
v'_i(j_\ell) = \begin{cases}
v_{ij}, \text{ if } i=\ell \text{ and } i\in [n]; \\
1, \text{ if } i \in \{ n + 2j - 1, n +  2j\}; \\
0, \text{ otherwise, }
\end{cases}
\]
i.e. each agent $i\in [n]$ values exactly one copy, $j_i$ for each $j\in \G$ at $v_{ij}$, and for each good $j\in \G$, there are exactly two dummy agents who value all copies of $j$.

We now state and use the following claim, and prove it immediately after the proof of Theorem~\ref{thm:lmprtopg}.
\begin{claim} \label{claim:differentsets}
Any MNW allocation $\x'$ of $\I'$ does not select two goods from same $S_j, j \in [m]$.
\end{claim}

Consider any MNW allocation $\x'$ of $\I'$. We construct a partition, $\x$ of goods for $\I$ from this in the following way. For $i\in[n]$, $j\in[m]$, define $x_{ij} = 1$ if $j_i \in \x'$, and 0 otherwise. Let $\x_i = \{j \in \G : x_{ij} = 1\}$. Thus, the value that agent $i$ gets in $\x$ is
\begin{equation*}
\begin{aligned}
v_i(\x_i) = \sum_{j\in \G} v_{ij} x_{ij} &= \sum_{j \in \G} v_{ij} \mathbf{1}(j_i \in \x')
= \sum_{j \in \G} v'_{i}(j_i) \mathbf{1}(j_i \in \x')
= v'_i(\x').
\end{aligned}
\end{equation*}
Thus, if $m \geq n$, $\NW(\x) = {\NW(\x')}^{(n+2m)/n}$ and the partition corresponding to $\x'$ as defined above gives an MNW solution for $\I$. On the other hand, if $m < n$, then $\x'$ already gives non-zero value to all dummy agents by Claim \ref{claim:differentsets}. Thus, to maximize the total number of agents who get non-zero value, it maximizes the number of agents in $[n]$ who get non-zero value. Call this set $S^*$. Thus partition $\x$ has maximum number of agents getting a non-zero value. Finally, it maximizes the Nash product over $S^* \cup \{n+1, \ldots, n+2m\}$. Claim \ref{claim:differentsets} also implies that all dummy agents get value $1$. Thus, $\prod_{i \in S^*} v_i(\x_i) = \prod_{i \in S^*} v_i(\x')$. Thus even in this case the allocation $\x$ corresponds to an MNW allocation in $\I$. 
\end{proof}

\begin{proof}[Proof of Claim \ref{claim:differentsets}]
Consider first $m \geq n$. Suppose $\exists j\in [m]$ for which two goods $j_i, j_{i'} \in \x', i \neq i'$. Since exactly $m$ goods are picked in $\x'$, there is some $j'\in [m]$, for which no good $j'_i$ is picked in $\x'$ for any $i\in [n]$. This implies that the agents $2j'+n-1, 2j'+n$ get zero value in $\x'$, making $\NW(\x') = 0$. However, choosing a good from each $j \in [m]$ gives non-zero value to all dummy agents. At the same time, since $m \geq n$, these goods can be chosen so that they give non-zero value to distinct agents in $[n]$. This makes $\NW(\x') \neq 0$ contradicting Nash optimality of $\x'$.

Now, if $m < n$ Nash welfare of all allocations in $\I$ is $0$. Thus, the MNW allocation is the one that maximizes the number of agents who get non zero value and then maximizes the product of values for these agents. Consider any allocation $\Bar{\x}$, suppose $\exists j\in [m]$ for which two goods $j_i, j_{i'} \in \Bar{\x}, i \neq i'.$ then again for some $j'$, agents $n+2j' - 1$ and $n+2j'$ get value 0 making $\NW(\Bar{\x}) = 0$. At the same time, even if $\Bar{\x}$ has goods from all different $S_j$, since $m < n$, and each one item from $S_j$ gives value only to one agent $i \in [n]$, the $\NW(\Bar{\x}) = 0$ even in this case. Thus, if $m < n$, all allocations have Nash welfare $0$ in $\I'$ also. Suppose the MNW allocation, $\x'$ had two goods from same $S_j$ for some $j \in [m]$. Then, there exists a $j' \in [m]$ such that no good is selected from $S_{j'}$. The two goods from $S_j$ give value to exactly four agents - the two dummy agents $2j+n-1, 2j+n$ and two agents who receive their copy of good $j$. Instead, if we exchange one of these goods to a good from $S_{j'}$, we give non-zero value to at least five agents - dummy agents $2j + n - 1, 2j + n, 2j'+n-1, 2j'+n$ and at least one of the agents in $[n]$. We did not change the value of any other agents in this process. Thus, we increase the number of agents who get non-zero value, contradicting the maximality of $\x'$. Thus, in both cases, all $m$ goods are picked from different $S_j, j \in [m]$. 
\end{proof}

\begin{observation}\label{rem:reductionvalues}
A desirable feature of the above reductions for the MNW problem from instance $\I = (\A,\G, V)$ to $\I'=(\A',\G',V')$ is that $V' = V \cup \{0,1\}$, i.e., the reduction only creates instances $\I'$ which have 0 and 1 as the only potentially additional values as compared to $\I$. We use this feature in establishing the computational complexity of computing an MNW allocation in the $\decision$ model with binary values, see Corollary~\ref{cor:pdmhardness}.
\end{observation}

We also show similar polynomial-time reductions between the three models for the problem of computing a leximin-optimal allocation.

\begin{restatable}{theorem}{thmlmpgtopd}\label{thm:lmpgtopd}
$\publiclex$ polynomial-time reduces to $\decisionlex$.
\end{restatable}
\begin{proof}
Let $\I = (\A,\G,k,\{v_i\}_{i\in \A})$ be an instance of the $\public$ model. For $k=m$, the leximin problem is trivial, since we can select all the $m$ goods. When $k \geq n$, we can construct an instance $\I' = (\A',\G',\{\G_j\}_{j\in \G'}\{v'_i\}_{i\in \A'})$ of the $\decision$ model from $\I$ in polynomial time, such that given a leximin allocation of $\I'$, we can compute a leximin allocation of $\I$ in polynomial time. To construct $\I'$, we first create a set $\A'$ of $n+2$ agents. The first $n$ agents here correspond to the $n$ agents in $\I$. The last $2$ agents are used in the construction, and ensure that exactly $k$ goods are selected in $I$.

We next create $m$ public issues: for each good $j\in \G$, we create an issue $j$ with two alternatives $(j,1)$ and $(j,2)$. That is, $\G' = [m]$, and $\G_j = \{(j,1),(j,2)\}$ for $j\in \G'$.

The valuations are as follows: for an agent $i\in \A'$, and an alternative $(j,c)$ of the issue $j \in \G'$, where $c\in \{1,2\}$:
\[
v'_i(j,c) = \begin{cases}
v_{ij}, \text{ if } c=1 \text{ and } i\in [n]; \\
\alpha\cdot (m-k), \text{ if } i = n+1 \text{ and } c = 1;\\
\alpha\cdot (k), \text{ if } i = n+2 \text{ and } c = 2;\\
0, \text{ otherwise. }
\end{cases}
\]
where $\alpha < 1/m^2$ is a sufficiently small constant. Essentially, each agent $i\in[n]$ values the `$1$' decisions of the issue $j\in \G'$ at $v_{ij}$, the agent $n+1$ values only the `$1$' decisions, and agent $n+2$ values only the `$2$' decisions.

Let $\x'$ be a leximin allocation for the instance $\I'$. Clearly $v'_i(\x') > 0$ for all agents, since there is some allocation that gives positive utility to all agents, and the minimum utility only improves in the leximin solution. In particular $v'_i(\x')_i \ge 1$ for all $i\in[n]$. For $c\in\{1,2\}$, let $S_c$ be the set of issues $j$ with decision $c$ in $\x'$. That is, $S_c = \{j\in [m] : \x'_j = c\}$. Let $k' = |S_1|$.  we note that $v'_{n+1}(\x') = \alpha (m-k)k'$, and $v'_{n+2}(\x') = \alpha k(m-k')$. Since $\alpha<1/m^2$, for each $i\in[n], b\in[2]$, we have $v'_{n+b}(\x') < v'_i(\x')$. Suppose $k'\neq k$. Then any allocation $\x''$ with $|\{j : \x''_j = 1\}|=k$ gives $v'_{n+1}(\x'') = v'_{n+2}(\x'') = \alpha k(m-k)$, which is a leximin improvement over $\x'$, since $\min(v'_{n+1}(\x''),v'_{n+2}(\x'')) > \min(v'_{n+1}(\x'),v'_{n+2}(\x'))$. Hence $k'= k$.

We now explain how we can relate $\x'$ of $\I'$ to the public goods instance $\I$. Intuitively, the decision $(j,1)$ corresponds to selecting the public good $j$, and $(j,2)$ corresponds to not selecting $j$. Let $\x = S_1 \subseteq \G$ be a set of public goods of cardinality $k'$. Then for any $i\in [n]$ we have that $v_i(\x) = v'_i(\x')$, since $v'_i(j,2) = 0$ for every $j\in [m]$. Further since $k'=k$, $|\x| = k$. Hence $\x$ is a feasible solution for $I$. Since for all $i\in [n]$, $v_i(\x) = v'_i(\x')$, $\x$ is a leximin allocation for $I$.

Since the number of agents and goods created in the reduction are polynomially many in the size of the instance $\I$, and all other computations can also be carried out in polynomial time, this is a polynomial time leximin-preserving reduction. 
\end{proof}

\begin{restatable}{theorem}{thmlmprtopg}\label{thm:lmprtopg}
$\privatelex$ polynomial-time reduces to $\publiclex$.
\end{restatable}
\begin{proof}
The proof follows from essentially the same reduction used to show Theorem~\ref{thm:mnw-priv-to-public}. 
\end{proof}

\section{Properties of MNW and Leximin}\label{sec:prop}
We prove that MNW and leximin-optimal allocations satisfy desirable fairness and efficiency properties in the $\public$ model as well. First, we show some interesting relations between our three fairness notions -- $\prop, \prop 1$, and $\rrs$ in the $\public$ model where $k\ge n$.\footnote{Note that when $k < n$, $\s{RRS}$ is 0. Any agent who gets $0$ value satisfies $\s{Prop}1$ when $k < n$ trivially. Thus, $\s{RRS}$ and $\s{Prop}1$ coincide when $k < n$. On the other hand, the proportional value will be non-zero even when $k = 1$ if the agent likes at least one good. Thus, there can be no multiplicative relation between $\s{RRS}$ and $\prop$ when $k < n$.} Our results are presented in the table below. 

\begin{table}[ht]
{
\centering
\begin{tabular}{ |c|c|c|c| } 
\hline
& $\s{RRS}$ & $\s{Prop}$ & $\s{Prop}1$ \\
\hline 
$\s{RRS}$ & \checkmark & $\frac{n}{2n-1}$ (Lem.~\ref{lem:rrs-prop})& \checkmark (Lem.~\ref{lem:rrs-prop1}) \\ 
$\s{Prop}$ & $1/n$ (Lem.~\ref{lem:prop-rrs}) & \checkmark & \checkmark \\ 
$\s{Prop}1$ & \xmark (Ex.~\ref{ex:prop1-rrs})& \xmark (Ex.~\ref{ex:prop1-rrs}) & \checkmark \\ 
\hline
\end{tabular}
\caption{Relations between the fairness notions for $k\ge n$. Each cell $(R,C)$ contains a factor $\alpha$ s.t. any allocation satisfying the row property $R$ implies an $\alpha$-approximation to the column property $C$. Cells with $\alpha = 1$ are marked with \checkmark, and with $\alpha = 0$ are marked with \xmark.}
}
\end{table}
\begin{restatable}{lemma}{lemrrspropone}\label{lem:rrs-prop1}
Any allocation that satisfies $\rrs$ also satisfies $\prop 1$.
\end{restatable}
\begin{proof}
Fix any agent $i$. Let $\x = \{h_1, h_2, \ldots, h_k\}$ be any allocation that satisfies $\rrs$. Let $\x_k^* = \{g_1, g_2, \ldots, g_k\}$ denote the top $k$ goods for agent $i$. We assume that the goods both in $\x$ and $\x_k^*$ are ordered in decreasing order of valuations according to agent $i$. Now, suppose that top $\ell$ goods of $\x$ match with top $\ell$ goods of $\x_k^*$, i.e. $v_i(h_j) = v_i(g_j), \forall j \leq \ell$ and $v_i(h_{\ell+1} < v_i(g_{\ell + 1}))$. Note that since $\x_k^*$ is the top $k$ goods of agent $i$, we cannot have that $v_i(h_j) > v_(g_j)$ for any $j \leq \ell$. We want to prove that $\rrs$ implies $\prop 1$. If $\x$ was already satisfying proportionality, it is obvious that $\x$ is $\prop 1$. If $\ell \geq d$, it is again easy to see that $\x$ is $\prop 1$. This is because, if $k = d$ then we already have top $k$ goods, giving a proportional allocation. If $k > d$, then we can remove any good from $h_{d+1}, \ldots, h_k$ and exchange it with $g_{d+1}$ to ensure proportionality, making the original allocation $\prop 1$. Finally, if $n$ divides $k$ then we have proportionality implied by $\rrs$ from Lemma \ref{lem:rrs-prop}.

Thus, we now assume that $\ell < d$, $k = nd+r$ with $r \leq n-1$ and that $\x$ is not already a proportional allocation. We know that $v(h_1, \ldots, h_\ell) = v(g_1, \ldots, g_\ell)$ and $v(h_1, \ldots, h_k) < \frac{1}{n}v(g_1, g_2, \ldots, g_k)$. Thus, 
\begin{equation}\label{eqn:upperbound}
v(h_{\ell+1}, \ldots, h_k) < \frac{1}{n} v(g_{\ell+1}, \ldots, g_k)
\end{equation}
Now, $v(h_k) \leq \frac{1}{k-\ell} v(h_{\ell+1}, \ldots, h_k)$. Thus,
\begin{equation}\label{eqn:smallgood}
    v(h_k) \leq \frac{1}{n \cdot (k - \ell)} v(g_{\ell+1}, \ldots, g_k)
\end{equation}
Now, consider the good $g_{\ell+1}$. It is the good with highest value that is not in $\x$. We prove that removing $h_k$ and adding $g_{\ell +1}$ gives us an allocation that is proportional.
Since $\ell < d$,
$v_i(g_{\ell+1}) \geq v_i(g_{nd+j}), \forall j \leq r$. Combining with the fact that $r < n$,
\begin{equation}\label{eqn:last_r}
   (n - 1) \cdot v_i(g_{\ell+1}) \geq v_i(g_{nd+1}, \ldots, g_{nd+r}).
\end{equation}
Again since the goods are arranged in decreasing order of valuations,
we have $v_i(g_1, \ldots, g_d) \geq v_i(g_{jd+1}, \ldots,g_{(j+1)d}), \forall 1 \leq j \leq (n-1) $. Thus,
\begin{equation}\label{eqn:except_topd}
    (n-1) \cdot v_i(g_1, \ldots, g_d) \geq v_i(g_{d+1}, \ldots, g_{nd}).
\end{equation}
Define, $LHS = (n - 1) v_i(g_{\ell+1}) + (n-1) v_i(g_1, \ldots, g_d)$. 
Combining \eqref{eqn:last_r} and \eqref{eqn:except_topd},
\begin{equation*}
\begin{aligned}
    LHS &\geq v_i(g_{nd+1}, \ldots, g_{nd+r}) +v_i(g_{d+1}, \ldots, g_{nd}) \\
    &= v_i(g_{d+1}, \ldots, g_k) \\
    &= v_i(g_{\ell + 1}, \ldots, g_{k}) -v_i(g_{\ell + 1}, \ldots, g_d)
\end{aligned}     
\end{equation*}
Thus we get,
\[
    (n-1)v_i(g_{\ell+1}) + (n-1)v_i(g_1, \ldots, g_{\ell}) \geq v_i(g_{\ell + 1}, \ldots, g_{k})-nv_i(g_{\ell + 1}, \ldots, g_d)
\]
Now adding $v_i(g_{\ell+1})$ on both sides and using $v_i(g_{\ell+1}) \geq \frac{1}{k - \ell} v_i(g_{\ell+1}, \ldots, g_k)$ gives:
\begin{align*}
    nv_i(g_{\ell+1}) &+ (n-1)v_i(g_1,\ldots,g_{\ell}) \\
    &\geq v_i(g_{\ell + 1}, \ldots, g_{k}) - nv_i(g_{\ell + 1}, \ldots, g_d) +\frac{1}{k-\ell} v_i(g_{\ell+1}, \ldots, g_k) \\
    &\geq  v_i(g_{\ell +1}, \ldots, g_k) -nv_i(h_{\ell + 1}, \ldots, h_k) + n v_i(h_k), 
\end{align*}
where the second inequality follows because $\x$ is $\rrs$ and from (\ref{eqn:smallgood}). Rearranging the above terms and using the fact that $v_i(g_1,\ldots,g_\ell) = v_i(h_1, \ldots, h_\ell)$, we get
\begin{align*}
    nv_i(g_{\ell+1}) + nv_i(h_1, \ldots, h_k) - nv_i(h_k) \geq v_i(g_1, \ldots, g_k)
\end{align*}
which implies that $\x$ is $\prop 1$. 
\end{proof}

\begin{restatable}{lemma}{lemrrsprop}\label{lem:rrs-prop}
Any allocation that is $\alpha$-$\rrs$ is also $\alpha \cdot \frac{n}{2n-1}$-$\prop$. Further, when $n$ divides $k$,  $\alpha$-$\rrs$ implies $\alpha$-$\prop$.
\end{restatable}

\begin{proof}
We will prove a stronger result assuming the valuations $\{v_i\}_{i \in \A}$ are \textit{monotone} ($v_i(S) \leq v_i(S \cup g)$ for all $S \subseteq \G$ and $g \in \G \setminus S$) and \textit{subadditive} (for all $S_1 \subseteq \G, S_2 \subseteq \G$, $v_i(S_1) + v_i(S_2) \geq v_i(S_1 \cup S_2)$). The class of subadditive valuations captures complement-free goods, and subsumes additive valuations.

Let $\x$ denote any subset of $k$ items that satisfies $\alpha \cdot \s{RRS}$. Fix any agent $i$. We have, 
\begin{align*}
    v_i(\x) \geq \alpha \cdot \max_{|\y| \leq \lfloor k/n\rfloor} v_i(\y).
\end{align*}
Let $\x^*$ denote the set of top $k$ goods of agent $i$. Let $k = nd+r$ where $0\le r < n$. We can partition $\x^*$ by dividing it into $n$ bundles, each of size $\lfloor k/n \rfloor$ and $r$ more bundles, each of size $1$. Note that when $k \geq n$, $\lfloor k/n \rfloor \geq 1$ and $r < n$. Thus, we get at most $2n-1$ bundles each of size at most $\lfloor k/n\rfloor$. We denote these bundles by $S_1, S_2, \ldots, S_l$, with $l \leq 2n-1$. Thus:
\begin{align}\label{eqn:lbundles}
    \nonumber v_i(\x^*) &= v_i(\cup_{i \in [l]}S_i), \\
    \nonumber&\leq \sum_{i \in [l]} v_i(S_i), \\
     &\leq \sum_{i \in [l]} \frac{1}{\alpha} \cdot v_i(\x), \\
    \nonumber &\leq v_i(\x) \cdot \frac{2n - 1}{\alpha}.
\end{align}
Here the second inequality follows from subadditivity and third follows because $\x$ is $\rrs$. Thus, we have:
\begin{align*}
    v_i(\x) \geq \frac{\alpha}{2n - 1}v_i(\x^*) = \alpha \cdot \frac{n}{2n-1}\s{Prop}_i.
\end{align*}
Further, when $n$ divides $k$, $r = 0$ and we get $l = n$ bundles each of size $k/n$. Thus, we have from (\ref{eqn:lbundles}):
\begin{align*}
    v_i(\x^*) \leq \frac{n}{\alpha} \cdot v_i(\x).
\end{align*}
In conclusion:
\[
    v_i(\x) \geq \frac{\alpha}{n} v_i(\x^*) = \alpha \s{Prop}_i,
\]
showing that $\alpha$-RRS implies $\alpha$-$\prop$.
\end{proof}

\begin{restatable}{lemma}{lemproprrs}\label{lem:prop-rrs}
Any allocation that satisfies $\alpha$-$\prop$ gives an $\alpha/n$ multiplicative approximation to $\rrs$, and this is tight.
\end{restatable}
\begin{proof}
Suppose a given allocation, $\x$ satisfies $\alpha$-$\prop$. Fix any agent $i$.
\begin{align*}
    v_i(\x) &\geq \alpha \cdot \frac{1}{n} \cdot \max_{|\y|\le k}  v_i(\y)
    \geq \alpha \cdot \frac{1}{n} \cdot \max_{|\y|\le \lfloor k/n \rfloor}  v_i(\y) = \frac{\alpha}{n} \cdot \s{RRS}_i.
\end{align*} 

\begin{example}[Tightness of Lemma \ref{lem:prop-rrs}]
Consider the following example. We have $n=2$ agents and $m = 5$ goods. Agent $1$ values goods $1$ and $2$ at $1$ each, does not value goods $3, 4, 5$. Agent $2$ values all goods at $1$. If $k=4$, the $\rrs$ value of agent $1$ is $2$. Her proportional value is $1$. Thus, picking goods $1,3,4,5$ gives agent $1$ her $\prop$ share but only ensures $1/n$ of her $\rrs$ share. \qedhere
\end{example} 
\end{proof}

\begin{example}[$\prop 1$ does not approximate $\prop$ or $\rrs$]\label{ex:prop1-rrs}
Finally, we note that a $\s{Prop}1$ allocation might not give an $\alpha$ approximation to $\s{RRS}$ for any $\alpha > 0$. Consider an instance of public goods allocation with $n=2$. We have $3$ goods. Agent $1$ values goods $1$, $2$ at value of $1$ and values good $3$ at 0. Agent $2$ values goods $1$, $2$ at $0$ and values good $3$ at $1$. If we want to select $k=2$ goods, then, selecting goods $1$ and $2$ gives agent $2$ value 0. This allocation is $\s{Prop}1$, but provides no multiplicative approximation to either $\rrs$ or $\prop$ for agent $2$.
\end{example}
\begin{remark}
We also note that this example not only provides a guarantee of $\prop 1$ but is $\prop 1$ and Pareto Optimal. Thus neither $\prop 1$ nor $\prop1$+PO give any multiplicative approximation to $\rrs$ or $\prop$. This also indicates that both the notions of fairness, leximin and MNW are strong because they provide $\prop 1$, $\s{PO}$ and multiplicative approximations to $\rrs$ and $\prop$.
\end{remark}

Next, we show that MNW allocations are fair: 
\begin{lemma}\label{lem:prop1}
All MNW allocations satisfy Prop1.
\end{lemma}
\begin{proof}
Suppose there exists an MNW allocation $\x$ that is not $\prop 1$. This implies for some agent $i\in\A$, for all pairs of goods $j\in\x$ and $j'\notin \x$, $v_i((\x\setminus j) \cup j') < \prop_i$. If $k < n$, $\prop_i \leq \max_{j \in \G} v_{ij}$, and swapping any good in $\x$ with this good will give her her proportional share. 

Consider now $k \geq n$. Since we assume each agent positively values at least one good, the MNW value is non-zero. Since MNW is scale-invariant, we scale the valuations of agents so that $v_h(\x) = 1$ $\forall h\neq i$. Let $g'$ be the highest-valued good of $i$ not in $\x$, i.e., $g' = \s{argmax}_{j \in \G \setminus \x} v_{ij}$. Let $\x_0 = \{j\in \x : v_{ij} < v_{ig'}\}$ be the set of goods in $\x$ that give $i$ strictly lesser value than $g'$. Since $i$ does not satisfy Prop1, $\x_0 \neq \emptyset$. Suppose we order the goods in $\G$ according to the valuation of $i$ as $\{g_1,\dots,g_m\}$, where $v_i(g_r) \ge v_i(g_s)$ for $1\le r\le s \le m$. Then $n\cdot\prop_i = v_i(g_1,\dots,g_k)$ by definition. Since $g'$ is the highest-valued good for $i$ not in $\x$, and further since every good in $\x_0$ is valued at less than $v_{ig'}$ by $i$, we can bound the total value to $i$ of the top $k$ goods $g_1,\dots,g_k$ as follows: $v_i(g_1,\dots,g_k) \le v_i(\x\setminus \x_0) + |\x_0|v_{ig'}$
which, using additivity of $v_i$, can alternatively  be written as:
\begin{equation}\label{eqn:prop-topk}
v_i(\x) + \sum_{j\in\x_0} (v_{ig'}-v_{ij}) \ge n\prop_i.
\end{equation}
Consider a good $g$ given by\footnote{\cite{PublicDecision} considered an \textit{issue} similarly}:
\[g \in \s{argmin}_{j\in\x_0} \frac{\sum_{h\in\A\setminus\{i\}} v_{hj}}{v_{ig'}-v_{ij}}.\]
Then by definition of $g$, we have:
\begin{equation}\label{eqn:mnw-prop1-2}
\begin{aligned}
\frac{\sum_{h\in\A\setminus\{i\}} v_{hg}}{v_{ig'}-v_{ig}} &\le \frac{\sum_{j\in\x_0}\sum_{h\in\A\setminus\{i\}} v_{hj}}{\sum_{j\in\x_0} v_{ig'}-v_{ij}}
&\le \frac{\sum_{h\in\A\setminus\{i\}} \sum_{j\in\x_0} v_{hj}}{n\prop_i-v_i(\x)} \\
&\le \frac{n-1}{n\prop_i-v_i(\x)}, \\
\end{aligned}
\end{equation}
where the first transition follows by rearranging terms in the numerator, and using \eqref{eqn:prop-topk} in the denominator, and the final transition follows by recalling that $v_h(\x) = 1$ for all $h\neq i$.

Let $\delta = v_{ig'} - v_{ig}$. We know $v_i(\x) + \delta < \prop_i$. Substituting this in \eqref{eqn:mnw-prop1-2}, and noting $\delta > 0$ gives:
\begin{equation}\label{eqn:mnw-prop1-3}
\frac{\sum_{h\in\A\setminus\{i\}} v_{hg}}{\delta} < \frac{1}{v_i(\x)+\delta}. 
\end{equation}
Let us now consider the allocation $\x' = (\x\setminus g)\cup g'$. We show $\NW(\x') > \NW(\x)$, thus contradicting the Nash optimality of $\x$. Since for any $h\neq i$, $v_h(\x') \ge v_h(\x)-v_{hg} = 1-v_{hg}$, we have:

\begin{equation*}
\begin{aligned}
\prod_{h\in\A}v_h(\x') &\ge v_i(\x')\prod_{h\in\A\setminus\{i\}} (1-v_{hg})
\ge(v_i(\x)+\delta)\bigg(1-\sum_{h\in\A\setminus\{i\}} v_{hg}\bigg) \\
&> (v_i(\x)+\delta)\bigg(1-\frac{\delta}{v_i(\x)+\delta}\bigg) = v_i(\x),
\end{aligned}
\end{equation*}
where the first transition uses Weierstrass' inequality \cite{weierstrass}, and the second transition uses \eqref{eqn:mnw-prop1-3}. This leads to $\NW(\x') > \NW(\x)$, giving the desired contradiction. Hence any MNW allocation satisfies Prop1. 
\end{proof}
Besides Prop1, the MNW allocation satisfies several other desirable properties, as our next result shows.
\begin{restatable}{theorem}{thmmnwmain}\label{thm:mnw-main}
All MNW allocations satisfy PO, $\prop 1$, and $1/n$-$\s{RRS}$. Further when $k\ge n$, MNW allocation implies $\frac{1}{2n-1}$-Prop.
\end{restatable}
\begin{proof}
If any MNW allocation did not satisfy Pareto optimality, then at least one of the agents gets a strictly higher value with values of all other agents not decreasing. Thus, if the MNW value is non-zero, we get an allocation with strictly higher Nash Product, contradicting the optimality of value of MNW. On the other hand, if MNW value is zero and the strict increase of value holds for one of the agents with non-zero value, then the Nash Product over these agents increases contradicting maximality of Nash Product of these agents. On the other hand, if the strict inequality holds for an agent who receives zero value, the number of agents with non-zero value increases, contradicting the maximality of number of agents who get non-zero value. In both cases, the optimality of MNW is contradicted. Thus any MNW allocation satisfies Pareto Optimality.

Next we prove that all MNW allocations satisfy $1/n$-$\rrs$. Suppose there exists an MNW allocation $\x$ that is not $1/n$-RRS. This implies that for some agent $i \in \A$, $v_i(\x) < \frac{1}{n}\rrs_i$. Let us order the goods according to $i$'s valuation: let $\G = \{g_1,g_2,\dots,g_m\}$, such that $v_i(g_r) \ge v_i(g_s)$, for all $1\le r\le s\le m$. Let $p = \lfloor \frac{k}{n} \rfloor$. When $k < n$, $p = 0$, in that case $\s{RRS}_i = 0$. Therefore, $k \geq n$. Observe that the round-robin share of $i$ is given by $\rrs_i = v_i(\{g_1,\dots,g_p\})$. We scale the valuations of the agents so that for every agent $i$, $v_i(\x) = 1$. In particular, this implies $\s{RRS}_i > n$. 

Let us order the goods in $\x$ according to $i$'s valuation: let $\x = \{j_1,j_2,\dots,j_k\}$, such that $v_i(j_r) \ge v_i(j_s)$, for all $1\le r\le s\le k$. Define for $r\in[p]$, $S_r = \{j_{rn-n+1},\dots,j_{rn}\}$, and $g'_r = \s{argmin}_{j\in S_r} \sum_{h\in \A \setminus \{i\}} v_{hj}$.

We now construct another allocation $\x'$ as follows. We first check if $g_1 \in \x$. If not, we begin constructing $\x'$ by removing $g'_1$ from $\x$ and adding $g_1$. If $g_1 \in \x$, then we proceed to check whether $g_2\in \x$ or not. For every $r\in [p]$, we remove $g'_r$ and add $g_r$ if $g_r$ is not in $\x$. If $g_r$ is already in $\x$ then for such an $r$ no operation is done. Since we are removing $g_r'$ and $v_i(g_r') < v_i(g_r) \leq v_i(g_s)$ for all $s < r$, this ensures that $\{g_1,\dots,g_p\} \subseteq \x'$, which shows $v_i(\x') \ge \s{RRS}_i > n$.
Observe that:
\begin{align*}
\sum_{r=1}^p \sum_{h\in \A \setminus \{i\}} v_h(g'_r) &\le \sum_{r=1}^p \frac{1}{n}\sum_{h\in \A \setminus \{i\}}\sum_{j\in S_r} v_{hj} &\text{ (def. of $g_r$)} \\
&\le\frac{1}{n} \sum_{r=1}^p \sum_{j\in S_r} \sum_{h\in \A \setminus \{i\}} v_{hj} &\text{ (rearranging)}\\
&\le \frac{1}{n} \sum_{j\in \x} \sum_{h\in \A \setminus \{i\}} v_{hj} \quad\quad&\text{ (def. of $S_r$)}\\
&\le \frac{1}{n} \sum_{h\in \A \setminus \{i\}} v_{h}(\x^*) \quad\quad\;&\text{ (rearranging)}\\
&= \frac{n-1}{n}.
\end{align*}
Then we have:
\begin{align*}
\NW(\x')^n &= \prod_{h \in \A}v_h(\x'), \\ 
&\geq v_i(\x') \prod_{h \in \A \setminus \{i\}} v_h(\x'), \\
&\geq v_i(\x') \prod_{h \in \A \setminus \{i\}} \bigg(1 - \sum_{r = 1}^p v_h(g'_r)\bigg), \\
&\ge v_i(\x')\bigg(1-\sum_{r=1}^p \sum_{h\in \A \setminus \{i\}} v_h(g_r)\bigg), \\
&> n\bigg(1-\frac{n-1}{n}\bigg), \\
&= \NW(\x)^n,
\end{align*}
which contradicts the fact that $\x$ is Nash optimal.

Combining this with Lemma (\ref{lem:rrs-prop}) and Lemma (\ref{lem:prop1}), we get the proof of the theorem. 
\end{proof}

\noindent Finally, we show similar fairness and efficiency properties for the leximin-optimal allocation. 

\begin{restatable}{theorem}{thmleximinmain}\label{thmleximinmain}
All leximin-optimal allocations are PO, satisfy $\rrs$ and $\prop 1$. Further, when $k \geq n$, a leximin-optimal allocation is also $(n/(2n-1))$-$\prop$.
\end{restatable}
\begin{proof}
It is easy to see that any leximin-optimal allocations will be PO since any Pareto-dominating allocation will also leximin-dominate. Also, leximin-optimal allocations satisfy $\rrs$ assuming we scale valuations so that $\rrs$ is 1 for all agents: If $k < n$, $\rrs = 0$ and the leximin-optimal allocation is obviously $\rrs$. When $k \geq n$, the round-robin algorithm gives each agent at least their most-preferred $\lfloor k/n \rfloor$ goods, i.e., their $\rrs$-value; and the leximin-optimal allocation must also give each agent a utility of at least their $\rrs$-value. Combining this with Lemmas~\ref{lem:rrs-prop1} and~\ref{lem:rrs-prop} completes the proof of the theorem.
\end{proof}

\section{Complexity of MNW and Leximin}\label{sec:hardness}
In this section, we show that $\publicmnw$ and $\publiclex$ are $\s{NP}$-hard. Our hardness results also hold for instances with binary values, which is in stark contrast to the private goods setting, where MNW and leximin-optimal allocations can be computed in polynomial-time. We defer most proofs of results in this section to Appendix~\ref{app:hardness}. Since the cases of $k\ge n$ and $k<n$ are interesting in their own right, we consider them separately. 

We show when $k<n$, that the Nash welfare objective cannot be approximated to any multiplicative factor in polynomial-time, unless $\s{P}=\s{NP}$.
\begin{restatable}{theorem}{kltn}\label{kltn}
Given a $\public$ allocation instance where $k<n$, computing an $\alpha$-approximation to MNW is $\s{NP}$-hard for any $\alpha>0$, even when all valuations are binary.
\end{restatable}
\begin{proof}
We reduce from Set Cover. The set cover problem takes as input a universe $\mathcal{U} = \{e_1, e_2, \ldots, e_n\}$ of $n$ elements, a family, $\mathcal{F} = \{F_1, F_2, \ldots, F_m\}$ of subsets of $\mathcal{U}$, i.e., $\mathcal{F} \subseteq 2^{\mathcal{U}}$. The problem asks to find the minimum set of subsets from $\mathcal{F}$ such that their union covers all $\mathcal{U}.$ It is well known that this problem is $\s{NP}$-hard \cite{karp1972reducibility}. To reduce this to an instance of MNW problem, create an agent $i$ for each  $e_i \in \mathcal{U}$. Corresponding to each $F_j \in \mathcal{F}, j \in [m]$, create good $g_j$ such that $v_i(g_j) = 1$ if and only if $e_i \in F_j$. Now, for any $k < n$, the MNW is non-zero if and only if there is a set cover of size $k$. This implies that we cannot differentiate between the case where MNW value is zero or non-zero making any multiplicative factor approximation in polynomial time impossible unless $\s{P}=\s{NP}$. 
\end{proof}

We first state a technical lemma proved in Appendix~\ref{app:hardness} that we need for the next proof.
\begin{restatable}{lemma}{lemamgm}\label{lem:amgm}
Given $n$  positive integers $a_1, \ldots, a_n$, with each $a_i \geq \ell$ for some $\ell \in \Z_+$, if $\sum_{i = 1}^n a_i = \ell \cdot n + r$, for some $r < n$, then the maximum value of $\prod_{i = 1}^n a_i$ is $(\ell + 1)^r \ell^{(n-r)}$ and is achieved if and only if $a_i = \ell + 1, i \in S$, $a_i = \ell, i \in [n] \setminus S$, where $S \subseteq [n]$ is any set such that $|S| = r$.
\end{restatable}

\begin{restatable}{theorem}{thmnpkgtn}\label{thm:np-k>n}
$\publicmnw$ is $\s{NP}$-hard, even when all valuations are binary.
\end{restatable}
\begin{proof}
We call the decision version of finding an MNW allocation as $\PGNSW$. An instance of $\PGNSW$ is given by $\I = (\A, \G, k, \{v_i\}_{i \in \A}, T)$. Here $\A$, $\G$, $k$ and $\{v_i\}_{i \in \A}$ are exactly as defined in the $\public$ instance. In addition we have an integer $T$. The problem is to \textit{decide} whether there is an allocation with Nash welfare at least $T$.

We reduce exact regular set packing $(\ERSP)$ to $\PGNSW$. In the input to this problem, there are $n$ elements $X=\{x_1,\dots,x_n\}$, family of subsets $\mathcal{F}=\{F_1,\dots,F_m\}$ where each $F_j \subseteq X$ and $|F_j|=d$. The problem is to compute a subfamily $\mathcal{F'}\subseteq \mathcal{F}$, $|\mathcal{F'}|=r$, s.t. for all $F_i\neq F_j \in \mathcal{F'}, F_i \cap F_j = \emptyset$. Let $\I = (X,\mathcal{F},d,r)$ be an instance of $\ERSP$.

We construct an instance, $\I' = \{\A, \G, k, \{v_i\}_{i \in \A}, T\}$ of $\PGNSW$ as follows. We create a set $\A = [n]$ of $n$ agents, a set $\G = \{g_1,\dots,g_m\}\cup \{d_1,\dots,d_n\}$ of $m+n$ public goods. For any agent $i \in \A$ and good $g_j \in \G$, $v_i(g_j) = 1$ if $x_i \in F_j$ else 0. For any agent $i \in \A$ and good $d_j \in \G$, $v_i(d_j)=1$. Finally we set, $ k = r+n$, and $T = ((n+1)^{dr}n^{n-dr})^{1/n}$. We claim that $\I$ is a yes-instance for $\ERSP$ iff $\I'$ is a yes-instance for $\PGNSW$. 

Let $\mathcal{F'}$ be any $\ERSP$ solution, so $|\mathcal{F'}|=r$. Then corresponding to every $F_j \in \mathcal{F'}$, we pick $g_j$. We also pick all goods $d_j$. We have thus picked $r+n$ goods. Each picked $g_j$ gives value to exactly $d$ agents, and no agent gets value from two different $g_j$'s because of set disjointness. So exactly $dr$ agents get value 1 from picked $g_j$'s and every agent gets a value of $n$ from the $d_j$'s. Thus the Nash welfare is exactly $((n+1)^{dr}n^{n-dr})^{1/n} = T$, as required by $\PGNSW$.

We now prove that if $\I'$ has an allocation $\x$ (of size $r+n$) with Nash welfare at least $T = ((n+1)^{dr}n^{n-dr})^{1/n}$ then $\I$ is a yes-instance for $\ERSP$. Suppose $\x$ does not include all goods $d_j$. From $\x$, we create an allocation $\x'$ by adding the goods $d_j$ not in $\x$ and removing an equal number of goods from $\x$ to maintain cardinality. Clearly, $\x'$ Pareto-dominates $\x$. Thus, we have:
\begin{align}\label{eqn:NashProduct}
    \NW(\x') > \NW(\x) = ((n+1)^{dr}n^{n-dr})^{1/n}.
\end{align}
Now, $\x'$ has all goods $d_j$ and $r$ other goods, each of which are liked by exactly $d$ agents. Thus, $\sum_{i=1}^n v_i(\x') = n^2 + dr$ and $v_i(\x') \geq n, \forall i \in [n]$. So from Lemma \ref{lem:amgm}, $\NW(\x') \leq ((n+1)^{dr}n^{n-dr})^{1/n}$ contradicting (\ref{eqn:NashProduct}). Thus $\x$ must have all goods $d_j$. Hence, $\sum_{i = 1}^n v_i(\x) = n^2 + dr$. Again, from Lemma (\ref{lem:amgm}), $\left(\prod_{i = 1}^{n}v_i(\x)\right)^{1/n} = T$ if and only if, $v_i(\x) = (n + 1)$ for $i \in S$ where $S \subseteq [n], |S| = dr$ and $v_i(\x) = n$ for all $i \in [n] \setminus S$. Thus, the goods in $\x$ give value to disjoint agents. This implies that the corresponding sets are disjoint, showing that $\I$ is a yes-instance. 
\end{proof}

Next, we show $\s{NP}$-hardness even when there are only two agents.
\begin{restatable}{theorem}{thmnpconst}\label{thm:np-const}
$\publicmnw$ is $\s{NP}$-hard, even for two agents.
\end{restatable}
\begin{proof}
We present a reduction from Equal Sized Partition, $\s{EQSP}$ to $\PGNSW$ problem with two agents. The standard Partition problem takes as input a set, $S = \{a_i\}_{i \in [|S|]}$ of non-negative integers, with $|S| = n$ and asks if there exist $S_1, S_2 \subseteq S$ such that $S_1 \cup S_2 = S, S_1 \cap S_2 = \emptyset$ and $\sum_{i \in S_1} a_i = \sum_{i \in S_2} a_i$. Partition problem has been proven to be $\s{NP}$-hard \cite{mertens2006easiest}. The equal sized Partition problem puts a further constraint of $|S_1| = |S_2|$.
This can be shown to be $\NP$-hard by reducing from Partition problem itself. We take an instance of the Partition and add $|S|$ zeroes to it. This new instance has a equal sized partition if and only if the original instance had a partition. 

We create an instance of $\PGNSW$ as follows, $\I = \{\A, \G, k, \{v_i\}_{i \in \A}, T\}$ as follows: $\A = [2], \G = [m]$. The value $v_{ij}$ of agent $i$ for good $j$ is given by:
\[
v_{ij} =
\begin{cases}
a_j + R, \text{ if } i=1 \\
C + R - a_j, \text{ if } i=2
\end{cases}
\]
where $R = \sum_{j\in[m]} a_j$. and $C = 2R/m$. Further we set $k=m/2$, and $ T = (R+Rm)/2$.

Initially, let $v_1(j) = a_j, j \in [m]$ and $v_2(j) = \frac{{2 \cdot \sum_{j = 1}^m a_j}}{m} - a_j$. Now, some of the $v_2(j)$ might become negative. Define $R = \text{max}\{-1 \times \min_{j \in [m]}v_2(j), 0\}$. The valuations of the agents are defined as $v_i(j) = R + v_i(j), i \in [2]$. Finally, we set $k = m/2$ and $T = \left( \frac{\sum_{j = 1}^m a_j}{2}+\frac{R \cdot m}{2} \right) $

We prove that $\s{EQSP}$ is a yes instance if and only if $\PGNSW$ is a yes instance.

$(\Rightarrow)$ If $\s{EQSP}$ is a yes instance, there is a set $S_1 \subseteq S$ such that $\sum_{j \in S_1} a_j = \sum_{j \in (S \setminus S_1)} a_j = R/2$ and $|S_1| = m/2$. We create an allocation $\x$ as follows : corresponding to each $a_j \in S_1$, add $j \in \x$. Thus, $v_1(\x) = v_2(\x) = R/2 + m\cdot R/2$ implying that $\NW(\x) = T$.

$(\Leftarrow)$ 
Suppose $\PGNSW$ is a yes instance. Therefore, there is an allocation $\x$ such that $\NW(\x) \geq T$ and $|\x| = m/2$. The value agent $1$ gets from $\x$ is 
\begin{equation*}
v_1(\x) = \sum_{i \in \x} a_i + \frac{R \cdot m}{2},
\end{equation*} 
and the value agent $2$ gets from $\x$ is
\begin{equation*}
v_2(\x) =  (C+R)\cdot \frac{m}{2} - \sum_{i \in \x} a_i. 
\end{equation*}
Therefore, the Nash product is $\left( \sum_{i \in \x} a_i + \frac{R \cdot m}{2} \right) \cdot \left( (C+R)m/2 - \sum_{i \in \x} a_i \right)$.
Note that since $C=2R/m$, the above expression takes its maximum value of  $\left(\frac{\sum_{j = 1}^m a_j}{2}+\frac{R \cdot m}{2}\right)^2$ when $\x$ is such that $\sum_{i \in \x} a_i = R/2$.

Thus, the $\NW(\x) \geq T$ if and only if there exists $\x$ such that $\sum_{i \in \x} a_i = R/2$. Then $(S_1,S_2)$ is a solution for $\s{EQSP}$ where $S_1 = \{a_j : j\in \x\}$, and $S_2 = S\setminus S_1$. 
\end{proof}

We next show a similar hardness results for computing leximin-optimal allocations, which as we show, apply even for instances with binary values. 

\begin{restatable}{theorem}{thmnphardlmbin}\label{thm:np-hard-lm-bin}
$\publiclex$ is $\NP$-hard, even when the valuations are binary.
\end{restatable}
\begin{proof}
We reduce from $c$-monotone SAT. An instance of $c$-monotone SAT is given by $n$ variables, $m$ clauses and a parameter $c$. The clauses in monotone SAT are restricted to have only positive literals. The question is to determine if there is a satisfying assignment with at most $c$ of the variables are set to true. We reduce this problem to an instance of public goods leximin as $\I = (\A, \G, k, \{v_i\}_{i \in \A})$ with $\A = [m+1]$, $\G = [n+m-c+1]$, $k= m+1$. We have $m$ agents corresponding to the $m$ clauses and one dummy agent. We have $n$ goods corresponding to the $n$ variables and $m-c+1$ extra dummy goods that all agents like, i.e., have value 1 for. An agent corresponding to a clause likes only the goods corresponding to the variables that appear in the clause. The dummy agent likes all the dummy goods and does not like any other goods. Thus, formally, for each variable $x_j, j \in [m]$, we have a good $j \in \G$. For each clause $c_i$, we have an agent $i \in \A$. We have for $i \leq m, j \leq n$, $v_i(j) = 1$ if $x_j \in c_i$. $v_i(j) = 1$ for all $i \in [m+1], m+1 \leq j \leq m+1-c$.
We claim that the monotone c-SAT is a yes instance if and only if every agent except the dummy agent receive a value of at least $m+2-c$ in the leximin allocation and the dummy agent gets a value $m+c-1$.

\noindent$\implies$ Suppose the monotone SAT is a yes instance. Thus, there are $c$ variables that can be set to $1$ and every clause is satisfied. Pick the goods corresponding to these $c$ variables and pick all the dummy goods. The dummy agent gets a value of $m-c+1$ and all other agents get a value of at least $m-c+2$.

\noindent $\impliedby$ We first observe that the any leximin-optimal allocation will always include all the dummy goods. Otherwise we could remove any good and add the dummy good that is not included. This does not reduce the value of any of the agents corresponding to clauses and strictly increases the value of the dummy agent, contradicting the leximin-optimality. Thus, all dummy goods are always included. This gives all agents a value of $m-c+1$ from $m-c+1$ dummy goods. We have to select further $c$ goods. If the leximin-optimal is giving at least $m-c+2$ value to all agents except the dummy agent, then we have $c$ goods which together give value at least $1$ to all the agents corresponding to clauses. Thus, we can set the variables corresponding to these goods to $1$ and this satisfies all the corresponding clauses. 
\end{proof}
\begin{remark}
We note that the above reduction without the dummy goods and dummy agent reduces to a public goods instance $\I = (\A, \G, k, \{v_i\}_{i \in \A})$ with $|A| = m$, $k = c$. Since in c-monotone SAT, $c < m$ (otherwise we can trivially answer YES), this proves the $\NP$-hardness of computing a leximin-optimal allocation for $k < n$.
\end{remark}

\begin{restatable}{theorem}{thmhardlmconst}\label{thm:hard-lm-const}
$\publiclex$ is $\s{NP}$-hard, even for two agents.
\end{restatable}
\begin{proof}
We in fact show a stronger result that finding an allocation that maximizes the minimum value itself is $\NP$-hard. Formally, we denote an instance of max-min problem as $\I = \{\A, \G, k,  \{v_i\}_{i \in \A}, T\}$. The problem is given a set $\A$ of agents, a set $\G$ of goods, we want to select $k$ goods such that the minimum value received by any agent is at least $T$. The proof follows using the same reduction as proof of Theorem \ref{thm:np-const}. We set $k = m/2$ as before and set $T = \frac{R+Rm}{2}$. We prove there is an allocation where every agent gets a value of at least $T$ if and only if $\s{EQSP}$ is a yes instance. If $\s{EQSP}$ is a yes instance, then by creating partition of goods as we did in Proof of Theorem \ref{thm:np-const} we get an allocation where both agents get value exactly $T$. On the other hand, suppose there exists an allocation, $\x$ where both agents receive value at least $T$ and $|\x| = m/2$. Then, agent $1$ has value \[
v_1(\x) = Rm/2+\sum_{i \in \x}a_i \ge T,
\]
from $\x$ and agent $2$ has value 
\[
v_2(\x) = (C+R)m/2 - \sum_{i \in \x} a_i \ge T.
\]
By assumption, $v_1(\x) \geq T$, thus we have $\sum_{i \in \x}a_i \geq T - Rm/2 = R/2$. Thus, we have $v_2(\x) = (C+R)m/2 - \sum_{i \in \x} a_i \leq (C+R)m/2 - R/2 = Rm/2 + R - R/2 = Rm/2+R/2 = T$. Thus $v_2(\x) \leq T$. Together with $v_2(\x) \geq T$ this means $v_2(\x) = T$, which implies $v_1(\x) = T$. Thus, $\sum_{i\in\x} a_i = R/2$, i.e., $\s{EQSP}$ is a yes instance. 
\end{proof}

Using the reductions of Theorems~\ref{thm:mnwpgtopd} and \ref{thm:lmpgtopd} and the $\s{NP}$-hardness results of this section, we obtain $\s{NP}$-hardness results for computing MNW and leximin allocations in the public decision making model. In fact, Observation~\ref{rem:reductionvalues} implies that this $\s{NP}$-hardness remains for the MNW problem even with the valuations are binary.

\begin{corollary}\label{cor:pdmhardness}
$\decisionmnw$ is $\s{NP}$-hard, even when all values are binary.
\end{corollary}

Using our reductions (Theorems~\ref{thm:mnwpgtopd} and \ref{thm:lmpgtopd}) together with the $\s{NP}$-hardness of $\publicmnw$ and $\publiclex$ (Theorems~\ref{thm:np-k>n} and \ref{thm:np-hard-lm-bin}) implies that:
\begin{corollary}
The problems $\decisionmnw$ and $\decisionlex$ are $\s{NP}$-hard.
\end{corollary}

\section{Algorithms for MNW and Leximin
}\label{sec:algorithm} 
In light of the above computational hardness, we turn to approximation algorithms and exact algorithms for special cases. We first present an algorithm that provides an $O(n)$ factor approximation to MNW and satisfies fairness properties of $\s{RRS}$, $\s{Prop}1$ when valuations $\{v_i\}_{i \in \A}$ are \textit{monotone} ($v_i(S) \leq v_i(S \cup g)$ for all $S \subseteq \G$ and $g \in \G \setminus S$) and \textit{subadditive} (for all $S_1 \subseteq \G, S_2 \subseteq \G$, $v_i(S_1) + v_i(S_2) \geq v_i(S_1 \cup S_2)$). The class of subadditive valuations captures complement-free goods, and subsumes additive valuations. Our algorithm assumes access to demand oracles\footnote{Subadditive valuations are set functions and cannot in general represented efficiently. We thus assume access to the functions through some oracles. Given a set of prices $p_j$ for each good $j \in \G$, a demand oracle returns any set $S$ that maximizes $v_i(S) - \sum_{j \in S} p_j$.} for the subadditive valuations. We use the following subroutine, $\s{Maximize}$, from~\cite{subadd} which takes:
\begin{itemize}
    \item Input: Set of goods, $\G$, the valuation function $v_i$ of the agent $i$, and an integer $r$; and returns:
    \item Output: $\x\subseteq \G$, s.t. $v_i(\x) \ge \frac{1}{2} \max_{S\subseteq \G, |S|\le r} v_i(S)$ 
\end{itemize}

\noindent Our algorithm, $\s{AlgGreedy}$, has two steps:
\begin{itemize} \label{algo}
    \item For all $i \in \A$, $\x_i \gets \s{Maximize}(\mathcal{G}, v_i, \lfloor{\frac{k}{n}}\rfloor)$
    \item Return $\x \gets \cup_{i \in \A} \x_i$
\end{itemize}
For additive valuations, we assume that $\s{Maximize}$ returns a set of $\lfloor k/n \rfloor$ most-preferred goods for each agent. This algorithm enables us to show that:

\begin{restatable}{theorem}{thmalgoapx}
There exists a polynomial-time algorithm for the problem of $\public$ allocation (where $k\ge n$ and agents have monotone, subadditive valuations) that returns an allocation which satisfies $\s{RRS}$, $\frac{1}{2}$-$\s{Prop}$, and approximates the MNW to a factor of $O(n)$. Further, when the valuations are additive, the allocation satisfies $\s{Prop}1$.
\end{restatable}
\begin{proof}
Let $\x = \cup_{i \in \A}\x_i$ be the output of $\s{AlgGreedy}$ and $\x^*$ be any MNW allocation for given instance $\mathcal{I}.$ Let $\x_{i, r}^*$ be the bundle of goods of size $r$ that maximize the valuation of agent $i$. We let $\x_1^*, \x_2^*, \ldots, \x_{2n-1}^*$ be any arbitrary partition of $\x^*$ with each part of size at most $\lfloor\frac{k}{n}\rfloor$. Note that such a partition is possible whenever $k \geq n$. We can write $k = n \cdot d + r$, with $r < n$ and thus create $r$ bundles each of size at most $1$ and $n$ bundles each of size at most $\lfloor k/n \rfloor$. Thus we have,
\begin{alignat*}{2}
    v_i(\x) &\geq v_i(\x_i) &\quad\text{(Monotonicity)} &\\
    &\geq \frac{1}{2}v_i(\x_{i, \lfloor\frac{k}{n}\rfloor}^*) & \quad\text{($\s{Maximize}$ subroutine)} & \\
    &\geq \frac{1}{2} \frac{1}{2n-1} \bigg(\sum_{j = 1}^{2n-1}v_i(\x_j^*)\bigg) \\ &\geq \frac{1}{2(2n-1)}v_i(\x^*). & \quad\text{(Subadditivity)}
\end{alignat*}
Finally, we prove the approximation for NW as follows.
\begin{equation*}
\begin{aligned}
\NW(\x) &= \bigg(\prod_{i = 1}^n  v_i(\x)\bigg)^\frac{1}{n}
\geq \bigg(\prod_{i = 1}^n \bigg( \frac{1}{2(2n-1)}v_i(\x^*)\bigg)\bigg)^\frac{1}{n}, \\
&= \frac{1}{2(2n-1)} \NW(\x^*).
\end{aligned}
\end{equation*}
We further note that $\s{AlgGreedy}$ by definition gives each agent her $\rrs$ share. Combining this with Lemma (\ref{lem:rrs-prop}) and Lemma (\ref{lem:rrs-prop1}), we get the Theorem.
\end{proof}

We now present pseudo-polynomial time algorithms for two special cases, namely constantly many types of agents, and constantly many types of goods. Our results apply to the more general model of \textit{budget constraints}. We denote an instance of this model by $\I = (\A, \G, B, \{c_j\}_{j \in \G}, \{v_i\}_{i \in \A})$. Each good $j \in \G$ has an associated integral cost $c_j$, and in a feasible allocation the sum of costs of the picked goods must not exceed the budget $B$. The MNW and leximin-objectives are defined as before, but over feasible allocations that satisfy the budget constraints. Since cardinality constraints are a special case of budget constraints with uniform cost, our hardness results apply for the budget model also.

\paragraph*{Constantly many types of agents.} We consider instances where the number of \textit{agent types} is constant. We say agents $i$ and $h$ have the same \emph{type} if $\forall j\in \G$, $v_{ij} = v_{hj}$. Using a dynamic-programming based algorithm, we prove the following theorem.
\begin{restatable}{theorem}{thmconstagents}\label{thm:constagents}
For a $\public$ allocation instance, $\I = (\A,\G,B, \{c_j\}_{j \in \G}, \{v_i\}_{i\in \A})$ with $t$ distinct types of agents, (i) an MNW allocation can be computed in time $O(m \cdot (mV)^t)$, (ii) a leximin-optimal allocation can be computed in time $O(m \cdot n \log n \cdot (mV)^t)$, where $V = \max_{i \in \A, j \in \G} v_{ij}$.
\end{restatable}
\begin{proof}
Suppose for the $t$ different types of agents, there are $w_\ell$ agents of type $\ell$, for $\ell\in[t]$. We rename $v_\ell$ to be the valuation function of agents of type $\ell$. In any allocation $\x$, all agents of the same type receive the same utility. Hence:

\[\NW(\x) = \bigg(\prod_{\ell\in[t]} (v_i(\x))^{w_\ell}\bigg)^{1/n}.\]

Note that the maximum value of any good is $V$. Thus, for any feasible allocation $\x$ and any agent $i$, $v_i(\x) \le mV$.

Our algorithm populates a table $T[u_1,\dots,u_t,j]$, for $0 \le u_i \le mV$ for every $i\in [t]$, and $j\in \G$. We store in $T[u_1,\dots,u_t,j]$ the lowest possible value $b$ such that there exists a subset $S$ of goods of cost at most $b$, which gives each agent $i$ a utility of $u_i$, i.e., $v_i(S) = u_i$ for all $i\in [t]$, and $j$ is the largest index item in $S$. Then we have:

\begin{equation}\label{eqn:const-agent-dp}
T[u_1,\dots,u_t,j] = c_j+\min\limits_{j'< j} T[u_1-v_{ij}, \dots, u_t-v_{tj},j'].
\end{equation}

Thus we can populate the table $T$ using dynamic programming. For the base case, we create a dummy good, $j = 0$ which gives utility $0$ to each agent and has cost 0. The size of table is $(mV+1)^t\cdot (m+1)$, which is pseudo-polynomial in the size of the instance $\I$, since $t$ is a constant. Together with the fact that at most $m+1$ cells need to be visited to compute the expression on the left in~\eqref{eqn:const-agent-dp}, we conclude that the entire table $T$ can be filled in pseudo-polynomial time.

To compute the MNW value, we iterate over all cells $T[u_1,\dots,u_t,j]$ of the table which satisfy $T[u_1,\dots,u_t,j] \le B$, and output the cell which maximizes $\prod_{i\in[t]} u_i^{w_i}$, which can again be done in pseudo-polynomial time. The allocation itself can be computed using standard techniques used in dynamic programming algorithms to keep track of the partial solution.

We note that using the same table, we can iterate over all cells that satisfy the budget constraint. We keep an initial candidate leximin-optimal allocation and then whenever we find a cell that satisfies the budget constraint, we sort the utility vector and compare it with our candidate leximin-optimal allocation. Thus, this takes time $O(m\cdot n \log n \cdot (mV+1)^t).$ 
\end{proof}

Theorem~\ref{thm:constagents} also implies:
\begin{corollary}
For binary valuations, with constantly many types of agents $\publicmnw$ and $\publiclex$ are polynomial-time solvable.
\end{corollary}

\paragraph*{Constantly many types of goods.} We now consider instances where the number of \textit{types of goods} is constant. We say two goods $j_1, j_2\in\G$ have same type if for all agents $i \in \A$, $v_{ij_1} = v_{ij_2}$ and $c_{j_1} = c_{j_2}$. In this case, we can enumerate all feasible allocations efficiently, implying that an MNW or leximin-optimal allocation can be computed in polynomial-time.
\begin{restatable}{theorem}{thmconstgoods}\label{thm:constgoods}
For a $\public$ allocation instance $\I = (\A, \G, B, \{c_j\}_{j \in \G}, \{v_i\}_{i \in [n]})$ with t different types of goods, (i) an MNW, can be computed in time $O(m ^t)$ (ii) a leximin-optimal allocation can be computed in time $O( n \log n \cdot m^t)$.
\end{restatable}
\begin{proof}
We prove this theorem by giving a pseudo-polynomial time algorithm. Suppose the number of distinct types goods of is $t$. We denote by $T_i$, goods of type $i\in[t]$, and let $c_i$ denote the cost of goods of type $i$. The algorithm populates a table $T[r_1, r_2, \ldots, r_t]$ where each $r_i$ denotes the number of items of type $T_i$ picked. This table has a size of $(m+1)^t$ which is polynomial in the size of instance $\I$ for constant $t$. Each cell in the table stores the $\NW$ value that we get when $r_i$ goods are picked of type $T_i$. To compute the optimal $\NW$ value, we iterate over all cells of the table that satisfy $\sum_{i \in [t]} r_i \cdot c_{i}\leq B$ and pick the maximum value. The corresponding allocation is also given by the index of the table.

In this same table, instead of NW values, we can store the utility vectors of the agents in each cell. We can then find the leximin-optimal allocation by iterating over all cells that satisfy $\sum_{i \in [t]} r_i \cdot c_{i}\leq B$. We keep a candidate leximin-optimal and whenever a cell satisfies the budget constraint, we sort the utility vector and compare it with our candidate leximin solution. Thus, the time required for this is $O((m+1)^t \cdot n \log n )$. 
\end{proof}

\section{Discussion}
In this paper, we considered the problem of allocating indivisible public goods to agents subject to a cardinality constraint. We showed fundamental connections between the models of private goods, public goods, and public decision making, by presenting polynomial-time reductions for the popular solution concepts of maximum Nash welfare (MNW) and leximin. We also showed that MNW and leximin-optimal allocations satisfy desirable fairness properties like Prop1 and RRS, and the efficiency property of PO. Further we showed that these objectives are computationally $\NP$-hard, including for several special cases like constantly many agents and binary valuations. Lastly, we designed an approximation algorithm for MNW and pseudo-polynomial time algorithms for the case of constantly many agents.

Our work opens up several interesting research directions. Firstly, extending our reductions to the budget model presents a challenging problem. A second question is devising an algorithm to compute a Prop1+PO or RRS+PO allocations in polynomial time, bypassing the hardness of computing MNW or leximin-optimal allocations. Appropriately defining properties like Prop1 in the budget model and investigating whether MNW and leximin satisfy them would be a third interesting research direction. Finally, designing constant-factor approximation algorithms, even for restricted cases like binary valuations, which captures a large class of voting-like scenarios, is another important open problem.



\bibliography{references}

\appendix

\section{Missing Proofs from Section~\ref{sec:hardness}}\label{app:hardness}
\lemamgm*
\begin{proof}
 Let $a_i = \ell + r_i$, $r_i \in [r] \cup \{0\}$ and $\sum_{i = 1}^n r_i = r$. Let $S_1 = \{ i : r_i \ge 1\}$ and $S_{0} = \{ i: r_i = 0 \}$. Since $\sum_{i=1}^n a_i = \ell \cdot n + r$,
\begin{equation}\label{eqn:partition}
|S_0| = n - r + \sum_{i \in S_1} (r_i - 1) .
\end{equation}
Now,
\begin{align*}
    \prod_{i=1}^n a_i &= \prod_{i \in S_1} (\ell + r_i) \cdot \prod_{i \in S_0} \ell,\\
    &= \prod_{i \in S_1} (\ell + r_i) \cdot \ell^{|S_0|},\\
    &= \prod_{i \in S_1} ((\ell + r_i) \cdot \ell^{r_i - 1}) \cdot \ell^{(n - r)}, \\
    &\leq \prod_{i \in S_1} (\ell + 1)^{r_i} \cdot \ell^{(n - r)},\\
    &=(\ell+1)^r \cdot \ell^{(n-r)}.
\end{align*}
where the third transition follows from (\ref{eqn:partition}), the fourth transition follows from binomial theorem and the last equality follows from (\ref{eqn:partition}) combined with the fact that $S_0, S_1$ form a partition of $[n]$. The fourth inequality is an equality if and only if $r_i = 1, \forall i \in S_1$. Thus the maximum value of product is $(\ell+1)^r \cdot \ell^{(n-r)}$. 
\end{proof}
\end{document}